\documentclass[11pt]{article}

\usepackage{amsmath,amssymb,amsbsy,amsfonts,amsthm,latexsym,
            amsopn,amstext,amsxtra,euscript,amscd,color}
            
\PassOptionsToPackage{hyphens}{url}\usepackage{hyperref}
    
\usepackage[capbesideposition=outside,capbesidesep=quad]{floatrow}

\newtheorem{thm}{Theorem}[section]
\newtheorem{lem}[thm]{Lemma}
\newtheorem{cor}[thm]{Corollary}

\newtheorem{prop}[thm]{Proposition}
\newtheorem{rem}[thm]{Remark}
\newtheorem{defn}[thm]{Definition}

\newcommand{\Tr}{{\rm Tr}}
\newcommand{\Trn}{{\rm Tr}_n}
\newcommand{\Trm}{{\rm Tr}_m}

\def\+{\oplus}

\def\cB{{\mathcal B}}
\def\cC{{\mathfrak C}}
\def\caC{{\mathcal C}}

\def\cW{{\mathcal W}}

\def\F{{\mathbb F}}
\def\Z{{\mathbb Z}}

\def\F{{\mathbb F}}

\def\Z{{\mathbb Z}}

\def\V{{\mathbb V}}

\def\xx{{\mathbf x}}

\def\00{{\bf 0}}
\def\11{{\bf 1}}
\def\+{\oplus}

\def\\{\cr}
\def\({\left(}
\def\){\right)}

\newcommand{\bwht}[2]{\mathcal{W}_{#1}(#2)}
\newcommand{\vwht}[3]{\mathcal{W}_{#1}(#2,#3)}

\newcommand{\cardinality}[1]{\# #1}

\usepackage[colorinlistoftodos,prependcaption,textsize=tiny]{todonotes}

\providecommand{\newoperator}[3]{%
  \newcommand*{#1}{\mathop{#2}#3}}

\newoperator{\FD}{\mathrm{FD}}{\nolimits}

\begin{document}
\title{\bf $C$-differential bent functions and perfect nonlinearity}
\author{Pantelimon~St\u anic\u a$^1$, Sugata Gangopadhyay$^2$, Aaron Geary$^1$, \and Constanza Riera$^3$, Anton Tkachenko$^3$
 \vspace{0.15cm} \\
 \small  $^1$ Applied Mathematics Department, \\
\small Naval Postgraduate School, Monterey, USA; \\
\small {\tt \{pstanica, aaron.geary\}@nps.edu}\\
\small $^2$ Department of Computer Science and Engineering,\\
\small Indian Institute of Technology Roorkee,
\small  INDIA;\\ 
\small \tt{sugata.gangopadhyay@cs.iitr.ac.in}\\
\small $^3$Department of Computer Science,\\
\small  Electrical Engineering and Mathematical Sciences,\\
\small    Western Norway University of Applied Sciences,\\
\small   5020 Bergen, Norway;  {\tt \{csr, atk\}@hvl.no}\\}

\maketitle

\begin{abstract}
 Drawing inspiration from Nyberg's paper~\cite{Nyb91}  on perfect nonlinearity and the $c$-differential notion we defined in~\cite{EFRST20}, in this paper we introduce 
 the concept of $c$-differential bent functions in two different ways (thus extending Kumar et al.~\cite{Ku85} classical definition). We further extend the notion of perfect $c$-nonlinear introduced in~\cite{EFRST20}, also in two different ways,   and show that, in both cases, the concepts of $c$-differential bent and perfect $c$-nonlinear are equivalent (under some natural restriction of the parameters). Some constructions of functions with these properties are also provided; one such construction provides a large class of PcN functions with respect to all $c$ in some subfield of the field under consideration. We also show that  both our classes of $0$-differential bents are supersets of  permutation polynomials, and that Maiorana-McFarland bent functions are not differential bent (of the first kind).
\end{abstract}


\section{Introduction and basic definitions}
 
 We will introduce here only some basic notations and definitions on Boolean and $p$-ary functions (where $p$ is an odd prime); the reader can consult~\cite{Bud14,CH1,CH2,CS17,MesnagerBook,Tok15} for more on these objects.

For a positive integer $n$ and $p$ a prime number, we denote by $\F_p^n$ the $n$-dimensional vector space over $\F_p$, and by $\F_{p^n}$ the  finite field with $p^n$ elements, while $\F_{p^n}^*=\F_{p^n}\setminus\{0\}$ will denote the multiplicative group. For $a\neq 0$, we often write $\frac{1}{a}$ to mean the inverse of $a$ in the multiplicative group of the  finite field under discussion. 
We use $\cardinality{S}$ to denote the cardinality of a set $S$ and $\bar z$, for the complex conjugate.
We call a function from $\F_{p^n}$ (or $\F_p^n$) to $\F_p$  a {\em $p$-ary  function} on $n$ variables. For positive integers $n$ and $m$, any map $F:\F_{p^n}\to\F_{p^m}$ (or, $\F_p^n\to\F_p^m$)  is called a {\em vectorial $p$-ary  function}, or {\em $(n,m)$-function}. When $p$ is fixed, we write $\V_{n,p}$ for the vector space $\F_{p^n}$, or $\F_{p}^n$ under consideration, and $\cB_{n,p}^m$, for the $p$-ary functions defined on $\V_{n,p}$ with values in $\V_{m,p}$. If $p=2$ we write $\V_n$ and $\cB_n^m$, and if $m=1$, we will drop the superscript, altogether.
When $m=n$, $F$ can be uniquely represented as a univariate polynomial over $\F_{p^n}$ (using some identification, via a basis, of the finite field with the vector space) of the form
$
F(x)=\sum_{i=0}^{p^n-1} a_i x^i,\ a_i\in\F_{p^n},
$
whose {\em algebraic degree}   is then the largest Hamming weight of the exponents $i$ with $a_i\neq 0$. To (somewhat) distinguish between the vectorial and single-component output, we shall use upper/lower case to denote the functions.
For a $p$-ary function $f:\F_{p^n}\to \F_p$,  the {\it Walsh-Hadamard transform} is defined as the complex-valued function
\[
\bwht{f}{u}  = \sum_{x\in \F_{p^n}}\zeta_p^{f(x)-\Trn(u x)}, \ u \in \mathbb{F}_{p^n},
\]
where $\zeta_q= e^{\frac{2\pi i}{q}}$, for any $q$, and $\Trn:\F_{p^n}\to \F_p$ is the absolute trace function, given by $\displaystyle \Trn(x)=\sum_{i=0}^{n-1} x^{p^i}$ (we will denote it by $\Tr$, if the dimension is clear from the context).  For $f\in\cB_{n,p}$, the map $\displaystyle\mathcal{F}_f(u ) = \sum_{\xx  \in \V_{n}} f(\xx)\zeta_p^{\Tr(ux)}$ is the Fourier transform of~$f$.
 The (vectorial) Walsh transform $\vwht{F}{a}{b}$ of an $(n,m)$-function $F:\F_{p^n}\to \F_{p^m}$ at $a\in \F_p^n, b\in \F_p^m$ is the Walsh-Hadamard transform of its component function ${\rm Tr}_m(bF(x))$ at $a$, that is,
\[
  \vwht{F}{a}{b}=\sum_{x\in\F_{p^n}} \zeta_p^{\Trm(bF(x))-\Trn(ax)}.
\]
NB: If one wishes to work with vector spaces, then one can replace the $\Tr$ by any scalar product on that environment, for example, if $\V_{n}=\F_p^n$, the vector space of the $n$-tuples over $\F_p$ we use the conventional dot product $u\cdot x$ for $\Tr(ux)$.
 
 In this paper, we will use both the absolute trace $\Trn$ and the relative trace $\Tr_{\F_{p^n}/\F_{p^m}}$, defined as $\Tr_{\F_{p^n}/\F_{p^m}} (x)=\sum_{i=0}^{\frac{n}{m}-1} x^{p^{mi}}$.
 
Given a $p$-ary  function $f\in\cB_{n,p}$, the derivative of $f$ with respect to~$a \in \F_{p^n}$ is the $p$-ary  function
$
 D_{a}f(x) =  f(x + a)- f(x), \mbox{ for  all }  x \in \F_{p^n}.
$

The sum $$\caC_{f,g}(z)=\sum_{x \in \V_{n}} \zeta_p^{f(x+z)  - g(x )}$$
is  the {\em crosscorrelation} of $f,g\in\cB_{n,p}$  at $z \in \V_{n}$.
The {\em autocorrelation} of $f \in \cB_{n,p}$ at $u  \in \V_{n}$
is $\caC_{f,f}(u )$ above, which we denote by $\caC_f(u )$.

For an $(n,m)$-function  
$F$, and $a\in\F_{p^n},b\in\F_{p^n}$, we let $\Delta_F(a,b)=\cardinality{\{x\in\F_{p^n} : F(x+a)-F(x)=b\}}$. We call the quantity
$\delta_F=\max\{\Delta_F(a,b)\,:\, a,b\in \F_{p^n}, a\neq 0 \}$ the {\em differential uniformity} of $F$. If $\delta_F= \delta$, then we say that $F$ is differentially $\delta$-uniform. If $m=n$ and $\delta=1$, then $F$ is called a {\em perfect nonlinear} ({\em PN}) function, or {\em planar} function. If  $m=n$ and $\delta=2$, then $F$ is called an {\em almost perfect nonlinear} ({\em APN}) function. It is well known that PN functions do not exist if $p=2$. While most of the literature deals with $(n,n)$-functions when it comes to differential uniformity, we see no reason why the concept (beyond its uses in $S$-boxes, of course) cannot be considered for all $(n,m)$-functions.

In~\cite{EFRST20} we  defined a multiplier differential and the corresponding difference distribution table   (in any characteristic).
 For an $(n,m)$-function $F$, and $a\in\F_{p^n},b\in\F_{p^m}$, 
 and $c\in\F_{p^m}$, the ({\em multiplicative}) {\em $c$-derivative} of $F$ with respect to~$a \in \F_{p^n}$ is the  function
\[
 _cD_{a}F(x) =  F(x + a)- cF(x), \mbox{ for  all }  x \in \F_{p^n}.
\]
We let the entries of the $c$-Difference Distribution Table ($c$-DDT) be defined by ${_c\Delta}_F(a,b)=\cardinality{\{x\in\F_{p^n} : F(x+a)-cF(x)=b\}}$. We call the quantity
\[
\delta_{F,c}=\max\left\{{_c\Delta}_F(a,b)\,|\, a\in \F_{p^n}, b\in\F_{p^m} \text{ and } a\neq 0 \text{ if $c=1$} \right\}\]
the {\em $c$-differential uniformity} of~$F$ (while we previously worked with $(n,n)$-functions, there is no reason why we should not consider general $(n,m)$-functions in this definition). We extend here for general $n$ and $m$ the concepts that,  in~\cite{EFRST20}, were defined for $m=n$:

 If $\delta_{F,c}=\delta$, then we say that $F$ is differentially $(c,\delta)$-uniform (or that $F$ has $c$-uniformity $\delta$, or for short, 
{\em $F$ has $\delta$-uniform $c$-DDT}). If $\delta=1$, then $F$ is called a {\em perfect $c$-nonlinear} ({\em PcN}) function (certainly, for $c=1$, they only exist for odd characteristic $p$; however, as proven in~\cite{EFRST20}, there exist PcN functions for $p=2$, for all  $c\neq1$). If $\delta=2$, then $F$ is called an {\em almost perfect $c$-nonlinear} ({\em APcN}) function. 
When we need to specify the constant $c$ for which the function is PcN or APcN, then we may use the notation  
$c$-PN, or $c$-APN.
It is easy to see that if $F$ is an $(n,n)$-function, that is, $F:\F_{p^n}\to\F_{p^n}$, then $F$ is $c$-PN if and only if $_cD_a F$ is a permutation polynomial.

The rest of the paper is organized as follows. 
Section~\ref{sec2} and \ref{sec3} introduce our two types of crosscorrelations/autocorrelations and define  (naturally) the concepts of perfect $c$-nonlinear and $c$-differential bent functions in the context of $(n,m)$-functions, and  show that $c$-differential bent functions correspond to perfect $c$-nonlinear functions (we use indices $1,2$ to specify which type of bentness or perfect nonlinearity we refer to).  Characterizations and some constructions of both concepts are provided. Section~\ref{sec4}  concludes the paper.

\section{The first crosscorrelation: $c$-differential bent$_1$ and perfect$_1$ $c$-nonlinear functions}
\label{sec2}

In this section we extend the PcN notion to allow arbitrary $p$-ary $(n,m)$-functions. We shall recover some results shown in~\cite{EFRST20, RS20} as particular cases.

As for the regular differentials, for $F\in\cB_{n,p}^m$ and fixed $c\in\V_m $, we define the {\em $c$-crosscorrelation} at $u\in\F_{p^n}, b\in\F_{p^m}$ by 
\[
{_c}\cC_{F,G}(u,b)=\sum_{x \in \F_{p^n}} \zeta_{p}^{\Trm(b(F(x+u)  - c G(x)))}
\]
and the corresponding {\em $c$-autocorrelation} at $z\in\F_{p^n}$, ${_c}\cC_{f}={_c}\cC_{f,f}$. Surely, ${_c}\cC_{F,G}(u,b)=\caC_{\Trm(bF),\Trm(bcG)}(u)$ 
and ${_c}\cC_{F}(u,b)=\caC_{\Trm(bF),\Trm(bcF)}(u)$ ($b$ can only be $0,1$ when $m=1$). 
We want to emphasize the $c$-differentials, which is going to be relevant later as it relates to the perfect $c$-nonlinear concept. (We do not want to complicate more the notation by using indices here, since it will be obvious which concept we refer to, because this autocorrelation has two input variables, while the second concept has only one input variable.)

Nyberg~\cite{Nyb91} extended the notion of perfect nonlinearity and called a function perfect nonlinear if its derivatives are balanced (i.e. they take every value the same number of times). Thus, the function's (non-trivial) autocorrelation must be zero. Likewise,  we now extend the definition of PcN, in the following way.
\begin{defn}
For arbitrary positive integers $m,n$, and $F$ an $(n,m)$-function and $c\in\F_{p^m}$ fixed,  we say that $F$ is {\em perfect$_1$ $c$-nonlinear} \textup{(}$PcN$\textup{)} if its $c$-autocorrelation 
${_c}\cC_F(u,b)=0$, for all $u\in \F_{p^n}^*$, $b\in\F_{p^m}^*$. 
A {\em strictly perfect$_1$ $c$-nonlinear} is a function $F$ for which all ${_c}\cC_F(u,b)=0$, 
for all $u \in \F_{p^n}$, $b \in \F_{p^n}^*$ \textup{(}obviously, strictly perfect$_1$ $c$-nonlinear functions do not exist for $c=1$\textup{)}. 
\end{defn}
NB: We removed $b=0$ from the domain, since in that case the autocorrelation of any function is constant,~$p^n$.

Surely, if the  
$c$-derivatives are balanced, that is, if
 ${_c}D_aF(x)=F(x+a)-cF(x)$, at every fixed $a\neq 0$, assumes  the same value $y\in\F_{p^m}$ for exactly  $p^{n-m}$ values of $x\in\F_{p^n}$,  
 then $F$ is perfect$_1$ $c$-nonlinear (similarly, at every fixed $a$ for  strictly perfect$_1$ $c$-nonlinear functions).
 Later we show that a function is perfect$_1$ $c$-nonlinear if and only if the traces of the $c$-differentials are balanced. It is clear that   PcN functions (for $m=n$) are strictly perfect$_1$ $c$-nonlinear functions, and of course, one wonders about the converse (again, for $n=m$). If all the traces of multiples of $c$-differentials are balanced and so,  for all $u\neq0$, the sum
 \[
 \sum_{x\in\F_{p^n}} \chi_b\left({_c}D_uF(x)\right)=0,
 \]
 for all $b\neq 0$, where $\chi_b(x)=\chi(bx)$ and $\chi$ is the canonical additive character of $\F_{p^n}$, then, by~\cite[Theorem 7.7]{LN97}, ${_c}D_uF (x)$ must be a permutation, hence $F$ is PcN.

 A known result for classical Boolean functions, was extended in~\cite{smgs} for generalized Boolean functions (that is, functions defined from $\V_n$ into $\mathbb{Z}_q$, where $q=2^k$), and a corresponding result connecting our definition of $c$-crosscorrelation to the Walsh transforms  of  general  $p$-ary functions, holds, as well. 
\begin{lem}
\label{lem1}
Let $p$ be a prime number and $m,n$ be nonzero positive integers. If $F,G\in\cB_{n,p}^m$ and $c\in\F_{p^m}$, then for all $b\in F_{p^m}$, we have
\begin{equation}
\begin{split}
\label{eq3}
& \sum_{ u  \in \F_{p^n}}{_c}\cC_{F,G}(u,b )\zeta_p^{-\Trn(ux)} =  \cW_F(x,b )\overline{\cW_{G}(x,bc)}, \\
& {_c}\cC_{F,G}(u,b ) = p^{-n}\sum_{x  \in \F_{p^n}}\cW_F(x,b )\overline{\cW_{G}(x,bc)}\zeta^{\Trn(ux)}.
\end{split}
\end{equation}
In particular, if $F = G$, then
\begin{align*}
&\sum_{ u  \in \F_{p^n}}{_c}\cC_{F}(u,b)\zeta_p^{-\Trn(ux)} =  \cW_F(x,b )\overline{\cW_{F}(x,bc )}\\
& {_c}\cC_F(u,b ) = p^{-n}\sum_{x  \in \F_{p^n}}\cW_F(x,b) \overline{\cW_{F}(x ,bc)}\zeta_p^{\Trn(ux)}.
\end{align*}
\end{lem}
\begin{proof}
We start with
\allowdisplaybreaks
\begin{align*}
&\sum_{ u  \in \F_{p^n}} {_c}\cC_{F,G}(u,b)\zeta_p^{-\Tr(ux)}
=  \sum_{\substack{u  \in \F_{p^n}\\ b  \in \F_{p^m}}} \sum_{z \in \F_{p^n}} \zeta_{p}^{\Trm(b(F(z+u)  - c    G(z ))}\zeta_p^{\Trn(-ux)}\\
=& \sum_{ u  \in \F_{p^n}} \sum_{z \in \F_{p^n}} \zeta_{p}^{\Trm(b(F(z+u)  - c    G(z ))}\zeta_p^{-\Tr((z+u)x)+ \Trn(zx)}\\
= &\sum_{z  \in \F_{p^n}}  \zeta_{p}^{-\Trm(b c  G(z ))} \zeta_p^{\Trn(zx)}\sum_{u \in \F_{p^n}} \zeta_{p}^{\Trm(bF(z+u)) }\zeta_p^{-\Trn((z+u)x)}\\
&\stackrel{w:=z+u}{=} \sum_{z  \in \F_{p^n}}  \zeta_{p}^{-\Trm(bc    G(z ))} \zeta_p^{ \Tr(zx)} \sum_{w \in \F_{p^n}} \zeta_{p}^{\Trm(bF(w)) }\zeta_p^{-\Trn(wx)}\\
=&\cW_F(x,b)\overline{\cW_{G}(x,bc)}.
\end{align*}
For the second identity, we reverse the argument, and obtain
\allowdisplaybreaks
\begin{align*}
&p^{-n}\sum_{x  \in \F_{p^n}}\cW_F(x,b )\overline{\cW_{G}(x,bc )}\zeta^{\Trn(ux)}\\
=&  p^{-n}\sum_{x  \in \F_{p^n}}  \sum_{z,w  \in \F_{p^n}}  \zeta_{p}^{- \Trm(bc   G(z ))} \zeta_p^{ \Trn(zx)} \zeta_{p}^{\Trm(bF(w)) }\zeta_p^{-\Trn(wx)} \zeta^{\Trn(ux)}\\
=& p^{-n} \sum_{z,w  \in \F_{p^n}}  \zeta_{p}^{\Trm(b(F(w)- c  G(z ))}  \sum_{x  \in \F_{p^n}} \zeta_p^{ \Trn((u+z-w)x)} \\
=& \sum_{z  \in \F_{p^n}}  \zeta_{p}^{\Trm(b(F(u+z)- c   G(z ))} ={_c}\cC_F(u,b).
\end{align*}
The claimed consequences are immediate.
\end{proof}

We know that the bent notion exists from any group $A$ to another group~$B$~\cite{Pott04}, defined via character theory.  There are many generalizations of the bent concept and we  mention here~\cite{HMP18,HP15,Ku85,MMMS17,Waifi18,mms0,gbps,MRS19,Me18,smgs,txqf,ZXSZ16, ZW11}.
For example, a $p$-ary function $F:\F_{p^n}\to \F_p$ is bent if the complex absolute value of the Walsh transforms is constant, namely, $|\cW_F(x)|^2=\cW_F(x)\overline{\cW_F(x)}=\caC_F(0)=p^{n}$, for all $x\in\F_{p^n}$. In that spirit, for $m\leq n$, we define a new bent concept below that takes into account the differential type used. 
\begin{defn}
We say that a function $F\in\cB_{n,p}^m$  is {\em $c$-differential bent$_1$} if $
\displaystyle \cW_F(x,b)\overline{\cW_{F}(x,bc)}={_c}\cC_F(0,b)$, for all $x\in\F_{p^n},b\in\F_{p^m}^*$.
\end{defn}

Below, we will show that a function $F\in\cB_{n,p}^m$ is  $c$-differential bent$_1$ if the traces of all of its $c$-differentials, ${_c}D_aF$ with $a\neq 0$, are balanced, thereby extending Nyberg's result~\cite{Nyb91} on perfect nonlinearity being equivalent to bentness for functions defined from $\F_{p^n}$ into $\F_p$. We can also regard it as an extension of the PcN property we defined (for $n=m$)  in~\cite{EFRST20}.

\begin{thm}
Let $1\leq m\leq n$ be  integers, $p$ prime, and $F\in\cB_{n,p}^m$, $1\neq c\in\F_{p^m}$. Then $F$ is perfect$_1$ $c$-nonlinear if and only if $F$ is  $c$-differential bent$_1$. Moreover, $F$ is strictly perfect$_1$ $c$-nonlinear if and only if $\displaystyle \cW_F(x,b)\overline{\cW_{F}(x,bc)}=0$, 
for all $x\in\F_{p^n},b\in\F_{p^m}^*$.
\end{thm}
\begin{proof}
We first assume that $F$ is perfect$_1$ $c$-nonlinear, and so, ${_c}\cC_F(u,b)=0$,  for all $u\in \F_{p^n}^*$ and $b\in\F_{p^m}^*$.
From Lemma~\ref{lem1}, for an arbitrary $b\in\F_{p^m}^*$, we compute
\allowdisplaybreaks
\begin{align*}
 &\cW_F(x,b)\overline{\cW_{F}(x,bc)}
 =\sum_{u  \in \F_{p^n}}{_c}\cC_{F}(u,b )\zeta_p^{-\Trn(ux)} \\
 = & {_c}\cC_F(0,b)+\sum_{0\neq u  \in \F_{p^n}} \zeta_p^{-\Trn(ux)}  {_c}\cC_F(u,b)\\
 =&{_c}\cC_F(0,b),
 \end{align*}
 where we used 
 the assumption that the $c$-autocorrelations ${_c}\cC_F(u,b)$ are zero, except,  possibly, at $u=0$.
 
 For the reciprocal, we  assume that $F$ is $c$-differential bent$_1$, that is, $\cW_F(x ,b) \overline{\cW_{F}(x,b c)}={_c}\cC_F(0,b)$, $b\neq 0$.   Then, for any $ b\in\F_{p^m}^*$ and $u\in\F_{p^n}^*$,
 \allowdisplaybreaks
 \begin{align*}
{_c}\cC_F(u,b ) &=  p^{-n}\sum_{x  \in \F_{p^n}}\cW_F(x ,b) \overline{\cW_{F}(x,b c)}\zeta_p^{\Trn(ux)}\\
& =p^{-n}{_c}\cC_F(0,b) \sum_{x  \in \F_{p^n}} \zeta_p^{\Trn(ux)}=0,
 \end{align*}
 where we used the same property that the exponential sum of a balanced function (in this case $\Tr(ux)$, for $u\neq0$) is zero. 
 This proves the first claim. The second claim follows easily using the equations above.
 \end{proof}

 We now discuss some of the differential properties of a perfect$_1$ $c$-nonlinear function.

 \begin{thm}
 \label{thm:bal_diff}
 Let $m,n$ be positive integers, $p$ a prime integer, $F\in\cB_{n,p}^m$, and $c\in\F_{p^m}$ fixed. Then $F$ is a perfect$_1$ $c$-nonlinear function \textup{(}$c$-differential bent$_1$\textup{)} if and only if, for all $b\neq 0,u\neq 0$ fixed, $x\mapsto \Trm(b(F(x+u)-cF(x))$ is balanced.
 \end{thm}
 \begin{proof}
  With $c\in\F_{p^n}$ constant, for every $u\in\F_{p^n}, b\in\F_{p^m}$, $0\leq j\leq p-1$, we let  $S_{j,c}^{u,b}=\{x\in\F_{p^n}\,|\, \Trm(b(F(x+u)-cF(x)))=j\}$.  We will use below that the order of the cyclotomic polynomial of index $p^m$ is  $\phi(p^m)=p^{m-1}(p-1)$.
 
 First, recall that the $p^k$-cyclotomic polynomial is $\phi_{p^k}(x)=1+x^{p^{k-1}}+x^{2p^{k-1}}+\cdots+x^{(p-1)p^{k-1}}$. In particular, we deduce that $\zeta_p^{p-1}=-(1+\zeta_p+\cdots+\zeta_p^{p-2})$.
 If   $u\in\F_{p^n}^*,  b\in\F_{p^m}^*$, and $F$ is perfect$_1$ $c$-nonlinear, then
 \allowdisplaybreaks
 \begin{align*}
0={_c}\cC_F(u ,b)
 &=\sum_{x \in \F_{p^n}} \zeta_{p}^{\Trm(b(F(x+u)  - c  F(x )))}\\
 &=\sum_{j=0}^{p-1} |S_{j,c}^{u,b}| \zeta_{p}^{j}
  =  \sum_{j=0}^{p-2}\left(|S_{j,c}^{u,b}|-  |S_{p-1,c}^{u,b}| \right) \zeta_{p}^{j}.
 \end{align*}
 The  extension $\mathbb{Q} \stackrel{p-1}{\hookrightarrow}\mathbb{Q}(\zeta_{p})$ has degree $p-1$ and the elements in following set $\left\{\zeta_{p}^j\,|\,0\leq j\leq p-2\right\}$  are   linearly independent in $ \mathbb{Q}(\zeta_{p})$ over  $\mathbb{Q}$, therefore the coefficients in the displayed expression are zero, that is, that for all $0\leq j\leq p-2$,
$ |S_{j,c}^{u,b}|=|S_{p-1,c}^{u,b}| $. Summarizing, for any $0\leq j\leq p-1$, the cardinality of the set $S_{j,c}^{u,b}$ is independent of~$j$, and so, for all $c,b,u\neq 0$ fixed, the function $x\mapsto \Trm(b(F(x+u)-cF(x))$ is balanced.
 
If $x\mapsto \Trm(b(F(x+u)-cF(x))$ is balanced, by reversing the argument, we find that $f$ is perfect$_1$ $c$-nonlinear.
 \end{proof}
 
 As a consequence, we can easily characterize the $0$-differential bent$_1$ functions.   
\begin{cor}
\label{cor:0diffbent1}
 Let $F\in\cB_{n,p}^m$. The following statements are equivalent: 
 \begin{itemize}
 \item[$(i)$] $F$ is a  $0$-differential bent$_1$ \textup{(}perfect$_1$ $0$-nonlinear\textup{)} function;
  \item[$(ii)$]  $\cW_{F}(0,b)=0$, for all $b\neq 0$;
   \item[$(iii)$] \textup{(}Under $m=n$\textup{)} $F$ is a permutation polynomial.
  \end{itemize}
\end{cor}
\begin{proof}
When $c=0$, for $u\neq 0$ fixed, the map $x\mapsto \Trm(b(F(x+u))$ is balanced if and only if  $x\mapsto \Trm(b(F(x))$ is balanced (since $x\mapsto x+u$ is a bijection on the input set $\F_{p^n}$). Under $m=n$, using~\cite[Theorem 7.7]{LN97}, this is equivalent to $F$ being a permutation polynomial.
\end{proof}
 Thus, if $m=n$ and  $F$ is a permutation of $\F_{p^n}$, then $F$ is 0-differential bent$_1$ (since in this case, $F$ is PcN for $c=0$~\cite{EFRST20}).
 We give below  another example of $c$-differential bent$_1$ functions on $\F_{p^n}$, for all $c\neq 1$.  $F(x)=x^{p^k}$, a linearized monomial on $\F_{p^n}$, we compute 
 \begin{align*}
 \Trn(b\left( {_c}D_aF(x)\right))&=\Trn\left(b(x^{p^k}+a^{p^k}-cx^{p^k} \right)\\
 &=\Trn\left((1-c)x^{p^k}\right)+\Trn(a)\\
 &=\Trn\left((1-c)^{p^{-k}} x\right)+\Trn(a),
 \end{align*}
 which is balanced, if $c\neq 1$. 
 Thus,  any linearized monomial is a (strictly) perfect$_1$ $c$-nonlinear function, for all $c\neq 1$. In fact, given any linearized polynomial $L$, for which $\Trn\left((1-c)^{p^{-k}} L(x)\right)$ is balanced, then $L$ is a (strictly) perfect$_1$ $c$-nonlinear function, for all $c\neq 1$. 
 Thus, this class of perfect$_1$ $c$-nonlinear functions is a superclass of linearized polynomials $L$ whose trace $\Trn((1-c)^{p^{-k}} L(x))$ is balanced, and, furthermore, when $c=0$, is a superclass of permutation polynomials.
 
Surely, the question is whether there are other examples. We ran a SageMath code and  found some (strictly) perfect$_1$ $c$-nonlinear ($c$-differential bent$_1$) functions on small dimensions that are not linearized polynomials.
For instance, $F(x)=x^{3}$ is    perfect$_1$ $0$-nonlinear   on $\F_{2^3}$;  $F(x)=x^{5}$ is (strictly) perfect$_1$ $0$-nonlinear on $\F_{2^3}$ and  (strictly)  perfect$_1$ $\{0,2\}$-nonlinear on $\F_{3^3}$; $F(x)=x^{21}$ is   perfect$_1$ $c$-nonlinear for all $c\neq 1$ in $\F_{3^4}$;  
 $F(x)=x^{15}$ is (strictly) perfect$_1$ $\{0,2\}$-nonlinear on $\F_{3^3}$.
 From our first two examples (and several more of that type), we see that the Gold function is not always $0$-differential bent for small values of $n$, and so, we wondered what  happens, in general. The answer is provided by~\cite{EFRST20,RS20} for the Gold function.
  However, we can show a more general result, which, as a consequence, implies also the behavior of the Gold function. We  could not adapt the  methods from~\cite{EFRST20} to show the theorem, so we provide here an alternative method that proves quite useful to show several results at once.
\begin{thm}
  Let $p$ be a prime number, $n$ a positive integer and $F(x)=x^d$,  a monomial function. If $\gcd(d,p^n-1)=1$, then $F$ is  $0$-differential bent$_1$. If  $\gcd(d,p^n-1)=2$, then $F$ is not $0$-differential bent$_1$.
\end{thm}
\begin{proof}
If $\gcd(d,p^n-1)=1$, then,
\[
\sum_{x \in \F_{p^n}} \zeta_{p}^{\Trm(\alpha x^d)}=
\sum_{x \in \F_{p^n}} \zeta_{p}^{\Trm(\alpha x)}=0,\text{ if $\alpha\neq 0$},
\]
using the fact that $x\to x^d$ is a permutation if $\gcd(d,p^n-1)=1$, so if $x$ covers $\F_{p^n}$, then $x^d$ does the same,
therefore showing the first claim.

To show the second claim, by Corollary~\ref{cor:0diffbent1}, if $F$ were $0$-differential bent$_1$, then $\sum_{x \in \F_{p^n}} \zeta_{p}^{\Trm(\alpha F(x))}=0$. Assuming $\gcd(d,p^n-1)=2$, then we have the identity between the following  Gaussian sums
\[
\sum_{x \in \F_{p^n}} \zeta_{p}^{\Trm(\alpha x^d)}=
\sum_{x \in \F_{p^n}} \zeta_{p}^{\Trm(\alpha x^2)}
\]
(we use here the fact that under $\gcd(d,p^n-1)=2$, then $\{x^d\,|\,x\in\F_{p^n}\}=\{x^2\,|\,x\in\F_{p^n}\}$, which can be seen by making the change of variable $x\mapsto x^{d/\gcd(d,p^n-1)}$).
By~\cite[Theorem 5.33]{LN97}, we know that if $f(x)=a_2x^2+a_1x+a_0\in\F_{p^n}$, $a_2\neq 0$, then the Gaussian sum
\[
\sum_{x\in\F_{p^n}} \chi\left({\Trn(f(x))}\right)=\chi\left(a_0-a_1^2(4 a_2)^{-1}\right) \eta(a_2) \sum_{y\in\F_{p^n}^*} \eta(y)\chi (y),
\]
where $\eta$ is the quadratic character of $\F_{p^n}$ and $\chi$ is a nontrivial additive character of $\F_{p^n}$. In our case 
 $\chi(u)=\zeta_p^{\Trn(u)}$, and so, 
 our previous displayed sum is equal to (using further~\cite[Theorem 5.15]{LN97})
 \[
 \sum_{x \in \F_{p^n}} \zeta_{p}^{\Trm(\alpha x^2)}
 =\begin{cases} 
   \eta(\alpha) (-1)^{n-1} p^{n/2} & {\rm if}\ p\equiv 1\pmod 4\\
     \eta(\alpha) (-1)^{n-1} i^n p^{n/2} & {\rm if}\ p\equiv 3\pmod 4.
 \end{cases}
 \]
 From this last identity, we see that we cannot have $ \sum_{x \in \F_{p^n}} \zeta_{p}^{\Trm(\alpha x^2)}=0$, if $\alpha\neq 0$, and so, $F$ cannot be  $0$-differential bent$_1$.
\end{proof}
The following are some important corollaries (we use~\cite[Lemma~9]{EFRST20}: if $p=2$, then $\gcd(2^k+1,2^n-1)=1$ and, if $p>2$, then $\gcd\left(p^k+1,p^n-1\right)=2$, when  $\frac{n}{\gcd(n,k)}$ odd; also, when $n$ is even, $k$ is odd,  $\gcd(n,k)=1$, then $\gcd\left(\frac{3^k+1}{2},3^n-1 \right)=2$).  Note that   Corollary~\ref{Gold} is also a consequence of~\cite[Theorem 10 $(ii)$]{EFRST20} and~\cite{RS20}.
\begin{cor}\label{Gold}
Let $n,k$ positive integers with $\frac{n}{\gcd(n,k)}$ odd and $F(x)=x^{p^k+1}$ be defined on $\F_{p^n}$, $p$ an odd prime. Then $F$ is not $0$-differential bent$_1$. If $p=2$, then $F$ is $0$-differential bent$_1$.
\end{cor} 
The Gold function is not the only function for which we have this type of result. The Coulter-Matthews~\cite{CS97} PN function is yet another example of a function that is not $0$-differential bent$_1$ (hence not perfect$_1$ $0$-nonlinear), under some conditions, and it is $0$-differential bent$_1$, under some other conditions (see~\cite{EFRST20,RS20} for a general result on the function $x\mapsto x^{\frac{p^k+1}{2}}$ and its differential uniformity).  
\begin{cor}
Let $n=2m\geq 2$, $k$ odd,  $\gcd(n,k)=1$ \textup{(}so, $\gcd \left(\frac{3^k+1}{2}, 3^n-1 \right) = 2$\textup{)}. Then $F(x)=x^{\frac{3^k+1}{2}}$ is not $0$-differential bent$_1$. If $n,k$ are such that $\gcd\left(\frac{3^k+1}{2},3^n-1 \right)=1$, then $F(x)=x^{\frac{3^k+1}{2}}$ is $0$-differential bent$_1$.
\end{cor}

We can generate classes of $(n,m)$-functions that are $c$-differential bent$_1$ in the following way. We take $G$ to be a PcN function on $\F_{p^n}$ with respect to $c\in\F_{p^m}$, a proper subfield of $\F_{p^n}$ (that is, $m<n$, $m\,|\,n$). We then define $F(x)=\Tr_{\F_{p^n}/\F_{p^m}}\left(G(x)\right)$.  First, observe that since $c\in\F_{p^m}$, then $\Trn \left({_c}D_aG(x)\right)= \Trm \left( {_c}D_aF(x)\right)$. 
Now, if ${_c}D_aG$ is a permutation (using our assumption), then ${_c}D_aF= \Tr_{\F_{p^n}/\F_{p^m}}\left({_c}D_aG\right)$ is balanced, and so is $b\left({_c}D_aF\right)$, for $b\neq 0$. We now use the fact that multiplication by $b\neq 0$ simply shuffles the output. What we mean is that with notations, ${\rm Ker}(\Tr_{\F_{p^n}/\F_{p^m}})=\{x\in\F_{p^n}\,|\, \Tr_{\F_{p^n}/\F_{p^m}}(x)=0\}$, and  $A_i=\{x\in\F_{p^n}\,|\, \Tr_{\F_{p^n}/\F_{p^m}}(x)=\alpha^i \}$, where $\F_{p^m}=\{0,\alpha^i\,|\,0\leq i\leq p^m-2\}$ ($\alpha$ is a primitive element of $\F_{p^m}$), then, writing $b=\alpha^{i_0}$, the partition corresponding to $b\Tr_{\F_{p^n}/\F_{p^m}}$ is now $A_0,A_{i+i_0\pmod{p^m-1}}$. Using this and the transitivity of the traces, then $ \Trm \left(b\left({_c}D_aF\right)\right)$ is also balanced.
 We record this in the next proposition.
 \begin{prop}
Let  $m\,|\,n$, $m<n$, and $p$ prime. If $G$ is PcN on $\F_{p^n}$ with respect to $c\in\F_{p^m}$, then $F(x)=\Tr_{\F_{p^n}/\F_{p^m}}(G(x))$ is $c$-differential bent$_1$.
 \end{prop}
 It is obvious that not all  $c$-differential bent$_1$ functions from $\F_{p^n}\to\F_{p^m}$ come from traces of permutations on $\F_{p^n}$ (we can see that by taking a trace function $F$ of a PcN $G$, as above, and then interchanging output points with the same trace output value). More precisely,  we take $\F_{p^m}=\{0,\alpha^i\,|\,0\leq i\leq p^m-2\}$ ($\alpha$ is a primitive element of $\F_{p^m}$) and random $A_i,A_j$, $i\neq j$, as above. We now  define $H(x)=F(x)$, unless $x\in A_1\cup A_2$, when $H(x)=\alpha^j$, if $x\in A_i$ and $H(x)=\alpha^i$, if $x\in A_i$. 

Classical (binary) bent functions do not transfer easily in this generalized bent context. To argue that claim, we next show that Maiorana-McFarland bents cannot be $c$-differential bent$_1$ for $c\neq1$. 
\begin{prop}
Let $n = 2m$. Let $F: \F_{2^n} \rightarrow \F_{2^m}$ be 
a (bent) Maiorana-McFarland $(n, m)$-function defined by
\begin{equation}
\label{vect-mmf-1}
F(x,y) = x\pi(y), \mbox{ for all } x, y \in \F_{2^m},
\end{equation}
where $\pi : \F_{2^m} \rightarrow \F_{2^m}$ is a 
permutation. Then  $F$ cannot be $c$-differential
bent$_1$  for $c\neq1$. 
\end{prop}
\begin{proof}
 As is customary, we identify $\F_{2^n}$ with 
$\F_{2^m}^2 = \F_{2^m} \times \F_{2^m}$. 
The Walsh-Hadamard transform of $F$ at 
$((u, v), b) \in \F_{2^m}^2 \times \F_{2^m}^*$ is 
\begin{equation}
\label{wht-mmf-1}
\begin{split}
W_F((u,v), b) &= \sum_{x \in \F_{2^m}} \sum_{y \in \F_{2^m}} (-1)^{\Trm(bF(x,y)) + \Trm(ux) + \Trm(vy) }\\
&= \sum_{x \in \F_{2^m}} \sum_{y \in \F_{2^m}} (-1)^{\Trm(bx\pi(y)) + \Trm(ux) + \Trm(vy) }. 
\end{split}
\end{equation}
Then
\allowdisplaybreaks
\begin{align*}
&W_F((u,v), b) W_F((u,v), bc)\\ 
&= \sum_{\substack{x_1 \in \F_{2^m} \\ y_1 \in \F_{2^m}}} 
\sum_{\substack{x_2 \in \F_{2^m} \\ y_2 \in \F_{2^m}}}
 (-1)^{\Trm(x_1(\pi_b(y_1)  + u) ) + \Trm(x_2(\pi_{bc}(y_2) + u)) + \Trm(v(y_1 +y_2)) } \\
&= \sum_{\substack{y_1 \in \F_{2^m} \\ y_2 \in \F_{2^m}}}  (-1)^{\Trm(v(y_1 +y_2))}
\sum_{x_1 \in \F_{2^m} }
 (-1)^{\Trm(x_1(\pi_b(y_1)  + u))  }\\
&\qquad \times\sum_{x_2 \in \F_{2^m}} (-1)^{ \Trm(x_2(\pi_{bc}(y_2) + u)) } \\
&= 2^{2m} \sum_{\substack{y_1 \in \F_{2^m} \\ y_2 \in \F_{2^m}}}  (-1)^{\Trm(v(y_1 +y_2))}
\delta_0(\pi_{b}(y_1) + u)) \delta_{0}(\pi_{bc}(y_2) + u))\\
&= 2^{2m} (-1)^{\Trm(v(\pi_{b}^{-1}(u) + \pi_{bc}^{-1}(u)))},
\end{align*} 
where $\pi_b(x) = b\pi(x)$ and $\pi_{bc}(x)= bc \pi(x)$, for all $x \in\F_{2^m}$. Since the product of the Walsh coefficients is not independent of $u,v$ for $c\neq1$, our claim is shown.
\end{proof}
We now give a class of Dembowski-Ostrom   (bilinear) polynomials on $\F_{2^n}$ that are $c$-differential bent$_1$ for all $c\neq 1$ (PcN) in some subfield of $\F_{p^n}$, from the known class of (bilinear) DO polynomials of~\cite{Blo01}. The next theorem provides a new class of PcN functions.
\begin{thm}
\label{thm:bent1}
Let $k$ be a divisor of the positive integer $n$ such that $k\geq 2$, $n/k$ is odd, and $\Tr_{\F_{2^n}/\F_{2^k}}$ be the relative trace of $\F_{2^n}$ over $\F_{2^k}$ \textup{(}recall that $\Tr_{\F_{2^n}/\F_{2^k}} (x)=\sum_{i=0}^{\frac{n}{k}-1} x^{2^{ki}}$\textup{)}. Then, for any $a\in\F_{2^k}\setminus\F_2$, the polynomials
\[
F(x)=x\left(\Tr_{\F_{2^n}/\F_{2^k}}(x)+ax\right)
\]
are $c$-differential bent$_1$ \textup{(}PcN\textup{)} on $\F_{2^n}$, for all $1\neq c\in\F_{p^k}$.
\end{thm}
\begin{proof}
For easy writing, we shall use $\Tr$ for $\Tr_{\F_{2^n}/\F_{2^k}}$ in the  proof.
The case of $c=0$ is contained in~\cite{Blo01}, though, our proof will provide an argument for all $c\neq 1$ at once. To show our claim, it will be sufficient to show that, for $c\neq 1$ fixed, the $c$-differentials ${_c}D_u F$ are permutations on $\F_{p^n}$. We will use the well-known fact (and easy to show by expanding the trace and using the ``freshman identity'' $(A+B)^p=A^p+B^p$ in characteristic $p$)  that $\Tr_{\F_{p^n}/\F_{p^k}}(x^p)=\left(\Tr_{\F_{p^n}/\F_{p^k}}(x)\right)^p$.
First, since $a$ and $\Tr(x)$ are in $\F_{2^k}$, we have
\[
\Tr(F(x))=\Tr(x\Tr(x))+\Tr(ax^2)=\Tr(x)^2+a\Tr(x)^2=(1+a)\Tr(x)^2,
\]
and 
\begin{align*}
\Tr(F(x+u)+cF(x))&=(1+a)\Tr(x+u)^2+(1+a) c \Tr(x)^2\\
&=(1+a)\left((1+c)\Tr(x)^2+\Tr(u)^2\right).
\end{align*}
By absurd, we assume that for some fixed $u$, there exist $x\neq y$ in $\F_{2^n}$ such that  ${_c}D_u F(x)= {_c}D_u F(y)$. 
Thus, applying the relative trace $\Tr$ to the identity $F(x+u)+cF(x)=F(y+u)+cF(y)$, we obtain
\begin{align*}
(1+a)\left((1+c)\Tr(x)^2+\Tr(u)^2\right)=(1+a)\left((1+c)\Tr(y)^2+\Tr(u)^2\right).
\end{align*}
Since $a\neq 1$ and $c\neq 1$, we then get  $\Tr(x)=\Tr(y)$.  Going back to 
$F(x+u)+cF(x)=F(y+u)+cF(y)$,   we get 
\begin{align*}
&(x+u)(Tr(x+u)+a(x+u))+cx(Tr(x)+ax)\\
=& (y+u)(Tr(y+u)+a(y+u))+cy(Tr(y)+ay),
\end{align*}
which, by labeling $T=\Tr(x)=\Tr(y)$ and $t=\Tr(u)$, becomes
\begin{align*}
&(x+u)(T+t)+a x^2+au^2+cxT+acx^2\\
=&(y+u)(T+t)+a y^2+au^2+cyT+acy^2.
\end{align*}
Simplifying, we obtain
\[
(x+y)((1+c)T+t)=a(c+1)(x+y)^2,
\]
and since $x\neq y$, $a\neq 0$, $c\neq 1$, we infer that $x+y=\frac{(1+c)T+t}{a(c+1)}\in\F_{2^k}$. But then $0=2T=\Tr(x+y)=(x+y)Tr(1)=\frac{n}{k} (x+y)=x+y$, since $\frac{n}{k} $ is odd, implying that $x=y$, a contradiction.
\end{proof}
\begin{rem}
We can easily find \textup{(}via SageMath\textup{)}  for small dimensions~$n$, even the first class of bilinear Dembowski-Ostrom polynomials of~\textup{\cite{Blo01}} \textup{(}the Gold case was already treated earlier\textup{)}, namely, the  permutations
$F_a(x)=x^{2^k+1}+ax^{2^{n-k}+1}$, where $d=\gcd(n,k)$, $\frac{n}{d}$ is odd and $a\neq g^{t(2^d-1)}$ for all integers $t$; or $G_a(x)=x^{2^{2k}+1}+a^{2^k+1} x^{2^k+1}+ax^2$, $n=3k$, and $a\neq g^{t(2^k-1)}$ for all integers $t$, give us $c$-differential bent$_1$ functions. For example, 
$F_\alpha$ \textup{(}$\alpha$ is a primitive element in the finite field $\F_{2^n}$ under discussion\textup{)} is $\{0,\alpha^3 + \alpha^2 + \alpha, \alpha^3 + \alpha^2 + \alpha+1\}$-differential bent$_1$ \textup{(}PcN\textup{)} on $\F_{2^5}$ \textup{(}we took here the primitive polynomial $x^5 + x^2 + 1$\textup{);}
$G_\alpha$   is $\{0,\alpha^3 + \alpha^2 + \alpha, \alpha^3 + \alpha^2 + \alpha+1\}$-differential bent$_1$ \textup{(}PcN\textup{)} on $\F_{2^6}$ \textup{(}with the primitive polynomial $x^6 + x^4 + x^3 + x + 1$\textup{)}.
\end{rem}

It is not surprising that one cannot extend this theorem to the odd characteristic. Kyureghyan and \"Ozbudak~\cite{KO12} showed that if $p$ is odd, $q=p^n$, and $n\geq 5$, then $F(x)= x(\Trn(x)-ax)$ cannot be planar (that is, for all $a\neq 0$, $F(x+a)-F(x)$ is a permutation), and if $a=1,2$, then it is planar on $\F_{q^3}$ (the necessity of this last result was shown in~\cite{Blo01}); in~\cite{Ya13} it was proved that the above function is also not planar for $n\geq 4$.  

We next ask the  question whether one can characterize the differential bentness of any DO polynomial and we have such an attempt below.
Suppose that $F: \F_{p^n}\rightarrow \F_{p^m}$ (arbitrary). The $c$-autocorrelation 
of $F$ at $u \in \F_{p^n}$ and $b \in \F_{p^m}$ is 
\begin{align}
{_c}\cC_{F}(u,b)&=\sum_{x \in \F_{p^n}} \zeta_{p}^{\Trm(b(F(x+u)  - c F(x)))} \nonumber\\
&=\sum_{x \in \F_{p^n}} \zeta_{p}^{\Trm(b(F(x+u) -F(x) +F(x) - c F(x)))} \label{gen_autocor1}\\
&= \sum_{x \in \F_{p^n}} \zeta_{p}^{\Trm(b(F(x+u) -F(x)))}
\zeta_{p}^{\Trm(b(1-c)(F(x)))}. \nonumber
\end{align}

Using the above observation, we can characterize some cases when Dembowski-Ostrom (DO) polynomials are (or are not) $c$-differential bent$_1$ (we do not see an easy way to modify our method~\cite{EFRST20} to show such a result, so we use a different technique).

For an $(n,n)$-function $F\in\cB_{n,p}^n$, we let $\Omega_{F,i}= \{x\,|\, \Trn(F(x))=i\}$ be the $i$-support of $F$, $0\leq i\leq p-1$. For a Dembowski-Ostrom polynomial $\displaystyle F(x)=\sum_{i,j=0}^{n-1} a_{ij} x^{p^i+p^j}$, we let $L_u(x)=A_{n-1}x+ A_{n-2}^p x^p+\cdots+  A_1^{p^{n-2}} x^{p^{n-2} }+A_0^{p^{n-1}}x^{p^{n-1}}$ be the {\em linearized companion polynomial} at $u\in\F_{p^n}^*$, where $A_i=\sum_{k=0}^{n-1}  u^{p^k}(a_{ik}+a_{ki})$.

\begin{thm}
Let $n\geq 2$, $c\in\F_{p^n}$ fixed, and $\displaystyle F(x)=\sum_{i,j=0}^{n-1} a_{ij} x^{p^i+p^j}$ be a Dembowski-Ostrom polynomial on $\F_{p^n}$, $p$ prime. The following statements hold\textup{:}
\begin{itemize}
\item[$(i)$] If, for some $u\in\F_{p^n}^*$, there exists $b\in\F_{p^n}^*$ such that $L_u(b)=0$, where $A_i=\sum_{k=0}^{n-1}  u^{p^k}(a_{ik}+a_{ki})$, and $\sum_{i=1}^{p-1} \cardinality\Omega_{b(1-c)F,i}< p^{n-1}$, then $F$ is not $c$-differential bent$_1$.
\item[$(ii)$] If for $u,b\in\F_{2^n}^*$, when either $L_u(b)\neq 0$ and $\displaystyle \sum_{i=1}^{p-1}\zeta_p^i\sum_{x \in \Omega_{b(1-c)F,i}}  \zeta_p^{\Trm(bD_uF(x))}=0$, or,  $L_u(b)=0$ and   $\displaystyle \sum_{i=1}^{p-1}\zeta_p^i\sum_{x \in \Omega_{b(1-c)F,i}}  \zeta_p^{\Trm(bD_uF(x))}=(-1)^{bF(u)}p^{n}$,  then $F$ is $c$-differential bent$_1$.
\end{itemize}
\end{thm}
\begin{proof}
From~\eqref{gen_autocor1}, we infer
\allowdisplaybreaks
\begin{align*}
{_c}\cC_{F}(u,b)&=\sum_{x \in \F_{p^n}} \zeta_{p}^{\Trn(b(F(x+u) -F(x)))}
\zeta_{p}^{\Trn(b(1-c)(F(x)))} \\
&=\sum_{i=0}^{p-1} \zeta_p^i \sum_{x \in \Omega_{b(1-c)F,i}} \zeta_{p}^{\Trn(b(F(x+u) -F(x)))}\\
&= \sum_{x \in \F_{p^n}} \zeta_{p}^{\Trn(bD_uF(x))} - \sum_{i=1}^{p-1}(1-\zeta_p^i) 
\sum_{x \in \Omega_{b(1-c)F,i}}  \zeta_{p}^{\Trn(bD_uF(x))}.
\end{align*}
Surely, for $u$ fixed, 
\begin{align*}
D_uF(x)&=\sum_{i,j=0}^{n-1} a_{ij} (x+u)^{p^i+p^j}+\sum_{i,j=0}^{n-1} a_{ij} x^{p^i+p^j}\\
&=\sum_{i,j=0}^{n-1} a_{ij} \left( u^{p^i} x^{p^j}+u^{p^j} x^{p^i}+u^{p^i+p^j}\right)\\
&=\sum_{i=0}^{n-1} \left(\sum_{k=0}^{n-1}  u^{p^k}(a_{ik}+a_{ki}) \right) x^{p^i}+\sum_{i,j=0}^{n-1}u^{p^i+p^j}.
\end{align*}
We let $A_i=\sum_{k=0}^{n-1}  u^{p^k}(a_{ik}+a_{ki})$  and $A=\sum_{i,j=0}^{n-1}u^{p^i+p^j}$.
We now use~\cite[Theorem 5.34]{LN97},
which states that for a polynomial $f(x)=A_r x^{p^r}+A_{r-1} x^{p^{r-1}}+\cdots+A_1 x^p+A_0x+A$, then
\[
\sum_{x\in\F_{p^n}} \chi_b(f(x))=
\begin{cases}
\chi_b(A) p^n &{\rm if}\ bA_r+b^p A_{r-1}^p+\cdots+ b^{p^{r-1} } A_1^{p^{r-1}}+b^{p^r}A_0^{p^r}=0\\
0 &{\rm otherwise},
\end{cases}
\]
where $\chi$ is a nontrivial additive character of $\F_{p^n}$ and $\chi_b(y)=\chi(by)$.
In our case,  $\chi(y)=\zeta_{p}^{\Trn(y)}$ and so,
\begin{align*}
\sum_{x \in \F_{p^n}} \zeta_{p}^{\Trn(bD_uF(x))}
=& \sum_{x \in \F_{p^n}} \zeta_{p}^{ \Trn\left(b\sum_{i=0}^{n-1} A_i x^{p^i}\right)}\\
=&
\begin{cases}
\zeta_{p}^{bF(u)} p^n\ &{\rm if}\ L_u(b)=0\\
0\ &{\rm otherwise}.
\end{cases}
\end{align*}
If the $i$-support of $\Trn(b(1-c)F(x))$ satisfies $\sum_{i=1}^{p-1} \cardinality\Omega_{b(1-c)F,i}< p^{n-1}$, we therefore find that for $b$  satisfying $bA_{n-1}+b^p A_{n-2}^p+\cdots+ b^{p^{n-2} } A_1^{p^{n-2}}+b^{p^{n-1}}A_0^{p^{n-1}}=0$, then
\begin{align*}
&\left|\sum_{x \in \F_{p^n}} \zeta_{p}^{\Trn(bD_uF(x))}-  \sum_{i=1}^{p-1}(1-\zeta_p^i)  \sum_{x \in \Omega_{b(1-c)F,i}} \zeta_{p}^{\Trn(bD_uF(x))}\right|\\
& \geq \left|\zeta_{p}^{bF(u)}\right|p^n-\left| \sum_{i=1}^{p-1}(1-\zeta_p^i) \right|\sum_{i=1}^{n-1}\cardinality \Omega_{b(1-c)F,i}>0,
\end{align*}
where we used the fact that  $p=\sum_{i=1}^{p-1}(1-\zeta_p^i) $,
and the first claim is shown. 

The second claim follows a similar approach since the autocorrelation now is zero, under the imposed conditions and the theorem is shown.
\end{proof}
\begin{rem}
We can impose different conditions on the DO polynomial~$F$ such that $F$ becomes $c$-differential bent$_1$, but they all become too technical and we prefer to just give the idea above.
\end{rem}
When $p=2$, the theorem above takes a slightly simpler form.
\begin{cor}
Let $n\geq 2$, $c\in\F_{2^n}$ fixed, and $\displaystyle F(x)=\sum_{i,j=0}^{n-1} a_{ij} x^{2^i+2^j}$ be a Dembowski-Ostrom polynomial on $\F_{2^n}$. The following statements hold\textup{:}
\begin{itemize}
\item[$(i)$] If, for some $u\in\F_{2^n}^*$, there exists $b\in\F_{2^n}^*$ such that $L_u(b)=0$, where $A_i=\sum_{k=0}^{n-1}  u^{p^k}(a_{ik}+a_{ki})$, and $|\Omega_{b(1-c)F,1}|<2^{n-1}$, then $F$ is not $c$-differential bent$_1$.
\item[$(ii)$] If for $u,b\in\F_{2^n}^*$, $L_u(b)\neq 0$ and $\displaystyle \sum_{x \in \Omega_{b(1-c)F,1}}  (-1)^{\Trm(bD_uF(x))}=0$, or, if $L_u(b)=0$ and   $\displaystyle \sum_{x \in \Omega_{b(1-c)F,1}}  (-1)^{\Trm(bD_uF(x))}=(-1)^{bF(u)} 2^{n-1}$,  then $F$ is $c$-differential bent$_1$.
\end{itemize}
\end{cor}


 \section{A second crosscorrelation: $c$-differential bent$_2$ and perfect$_2$ $c$-nonlinearity}
  \label{sec3}
  
 In this section, we take a novel  route and define a (semi-vectorial) Walsh transform (and a crosscorrelation below) by identifying {\em only} the output domain (via some basis, generated by the primitive element $\alpha$)  with $\Z_{p^m}$, using the invertible map $\sigma:\F_{p^m}\to \Z_{p^m}$, $\sigma(a_0+a_1\alpha+\cdots+a_{m-1}\alpha^{m-1})=a_0+a_1p+\cdots+a_{m-1}p^{m-1}$ (the invertibility comes from the unique representation of an integer in base $p$). We now define the (semi-vectorial) {\em Walsh transform} by (we avoid writing $\sigma(F)$ and just use $F$ below, with the understanding that the exponent of $\zeta_{p^m}^{F(x)}$ has the meaning that we regard it in~$\Z_{p^m}$)
 \[
  \cW_F(a)=\sum_{x\in\F_{p^n}} \zeta_{p^m}^{F(x)} \zeta_p^{-\Trn(ax)}.
\]

As for the regular differentials, for $f\in\cB_{n,p}^m$ and fixed $c\in\V_m $, we define the {\em $c$-crosscorrelation} at $z\in\F_{p^n}$ by
\[
{_c}\cC_{F,G}(z)=\sum_{x \in \F_{p^n}} \zeta_{p^m}^{F(x+z)  - c   G(x )}
\]
and the corresponding {\em $c$-autocorrelation} at $z\in\F_{p^n}$, ${_c}\cC_{F}={_c}\cC_{F,F}$. Surely, if $m=1$, ${_c}\cC_{F,G}=\caC_{F,cG}$ and ${_c}\cC_{F}=\caC_{F,cF}$.  The proof of the following lemma is similar to the one of Lemma~\ref{lem1}, so we omit it.
 
\begin{lem}
\label{lem2}
Let $p$ be a prime number and $m,n$ be nonzero positive integers. If $F,G\in\cB_{n,p}^m$ and $c\in\F_{p^m}$, then 
\begin{equation}
\begin{split}
\label{eq3}
& \sum_{u  \in \F_{p^n}}{_c}\cC_{F,G}(u )\zeta_p^{-\Tr(ux)} =  \cW_F(x )\overline{\cW_{cG}(x )}, \\
& {_c}\cC_{F,G}(u ) = p^{-n}\sum_{x  \in \F_{p^n}}\cW_F(x )\overline{\cW_{cG}(x )}\zeta^{\Tr(ux)}.
\end{split}
\end{equation}
In particular, if $F= G$, then
\begin{align*}
&\sum_{u  \in \F_{p^n}}{_c}\cC_{F}(u )\zeta_p^{-\Tr(ux)} =  \cW_F(x )\overline{\cW_{cF}(x )}\\
& {_c}\cC_F(u ) = p^{-n}\sum_{x  \in \F_{p^n}}\cW_F(x ) \overline{\cW_{cF}(x )}\zeta_p^{\Tr(ux)}.
\end{align*}
\end{lem}

As before,  we define a perfect nonlinear and bent property that takes into account this type of autocorrelation and differential. 
\begin{defn}
For $m\leq n$,  we say that a function $F\in\cB_{n,p}^m$  is {\em $c$-differential bent$_2$} if $\displaystyle \cW_F(x)\overline{\cW_{cF}(x)}={_c}\cC_F(0)\ \forall x\in\F_{2^n}$.
\end{defn}
\begin{defn}
We say that $F$ is {\em perfect$_2$ $c$-nonlinear}  if its $c$-autocorrelation 
${_c}\cC_F(u)=0$, $u\in \F_{p^n}^*$. If, in addition,  ${_c}\cC_F(0)=0$, then   $F$ is {\em strictly perfect$_2$ $c$-nonlinear}.
\end{defn}

Below, we will show that a function $F\in\cB_{n,p}^m$ is  $c$-differential bent$_2$ if and only if $F$ is perfect$_2$ $c$-nonlinear  ($m\leq n$), thereby extending Nyberg's result~\cite{Nyb91}, in this context, as well. If $F$ is a PcN $(n,n)$-function (as we defined it in~\cite{EFRST20}), then $F$ is strictly perfect$_2$ $c$-nonlinear, since $F(x+u)-cF(x)$ is a permutation and ${_c}\cC_F(u)=0$. However, the reciprocal may not be true, in general, since a sum of powers of roots of unity being zero does not imply our uniform distribution of the exponents.

\begin{thm}
Let $1\leq m\leq n$ be  integers, $p$ prime, and $F\in\cB_{n,p}^m$, $1\neq c\in\F_{p^m}$. Then $F$ is perfect$_2$ $c$-nonlinear if and only if $F$ is  $c$-differential bent$_2$. In particular, $F$ is strictly perfect$_2$ $c$-nonlinear if and only if $\displaystyle \cW_F(x)\overline{\cW_{cF}(x)}=0$.
\end{thm}
\begin{proof}
We first assume that $F$ is perfect$_2$ $c$-nonlinear, and so, ${_c}\cC_F(u)=0$, for all $u\in \F_{p^n}^*$.
It is easy to see that, by using Lemma~\ref{lem1}, then
\allowdisplaybreaks
\begin{align*}
 \cW_F(x )\overline{\cW_{cF}(x )}
 &=\sum_{u  \in \F_{p^n}}{_c}\cC_{F}(u )\zeta_p^{-\Tr(ux)} ={_c}\cC_F(0).
 \end{align*}
 
 For the reciprocal, we  assume that $F$ is $c$-differential bent$_2$.   Then, for any $0\neq u\in\F_{p^n}$,
 \allowdisplaybreaks
 \begin{align*}
{_c}\cC_F(u ) &=  p^{-n}\sum_{x  \in \F_{p^n}}\cW_F(x ) \overline{\cW_{cF}(x )}\zeta_p^{\Tr(ux)}\\
& =p^{-n}{_c}\cC_F(0) \sum_{x  \in \F_{p^n}} \zeta_p^{\Tr(ux)}=0,
 \end{align*}
 where we used the property that the exponential sum of a balanced function (in this case $x\mapsto \Tr(ux)$, $u\neq 0$) is zero. 
 \end{proof}
 As for the $0$-differential bent$_1$, we can easily characterize $0$-differential bent$_2$.
\begin{cor}
 Let $F\in\cB_{n,p}^m$, with integers $m,n$, both greater than $1$. Then $F$ is a  $0$-differential bent$_2$ \textup{(}perfect$_2$ $0$-nonlinear\textup{)} function if and only if
$\cW_F(0)=0$. 
\end{cor}
\begin{proof}
For $c=0$, ${_c}\cC_F(0)=\cW_F(0)$.  
 Since $F$ is a  $0$-differential bent$_2$, then $\displaystyle \cW_F(a)\overline{\cW_{{\bf 0}}(a)}={_0}\cC_F(0)=\cW_F(0)\ \forall a\in\F_{2^n}$. However, $\displaystyle \cW_{{\bf 0}}(a)=\sum_{x\in\F_{p^n}} \zeta_{p}^{-\Trn(ax)}$, which is $0$, if $a\neq 0$, and $p^n$ if $a=0$. Thus, ${_0}\cC_F(0)=\cW_F(0)=0$. Conversely, assuming $\cW_F(0)=0$, the identity   $\displaystyle \cW_F(a)\overline{\cW_{{\bf 0}}(a)} =0$ will hold for all $a\neq 0$ (if  $a\neq 0$, then $\cW_F(a)$ is arbitrary). If $a=0$, then $\cW_F(0)=\sum_{x \in \F_{p^n}} \zeta_{p^m}^{F(x)}=\cW_F(0)=0$.
\end{proof}
 
 By the previous corollary, however, we find that if $m=n$, and $F$ is a permutation on $\F_{p^n}$, then $F$ is always going to be  $0$-differential bent$_2$ (since $\cW_F(0)=0$, if $F$ is a permutation).
 Surely, if $F$ is a permutation, then $F$ is clearly a $0$-differential bent$_2$  (perfect$_2$ $0$-nonlinear) function.
 Moreover, if $L$ is a linearized permutation polynomial on $\F_{p^n}$, then $L$ is a perfect$_2$ $c$-nonlinear function, for all $c\neq 1$. 
 To check that, we compute the autocorrelation of $L$, and get
 \[
 {_c}\cC_L(a)=\sum_{x\in\F_{p^n}} \zeta_{p^n}^{(1-c)L(x)+L(a)}=0,
 \]
 for all~$c\neq 1$, when $L$ is a permutation.
Thus, one can regard the set of $c$-differential bent$_2$ as a superclass of linearized permutation polynomials.
We summarize this discussion below.
\begin{prop}
If $L$ is a linearized permutation polynomial on $\F_{p^n}$, then $L$ is $c$-differential bent$_2$ \textup{(}perfect$_2$ $c$-nonlinear function\textup{)}, for all $c\neq 1$. If $m=n$, and $F$ is a permutation on $\F_{p^n}$, then $F$ is  $0$-differential bent$_2$ \textup{(}perfect$_2$ $0$-nonlinear function\textup{)}.
\end{prop}
 By SageMath, we get other polynomials. For example, $F(x)=x^3$ is perfect$_2$ $0$-nonlinear on $F_{2^3}$; 
$F(x)=x^3+x^5$ is perfect$_2$ $0$-nonlinear on $F_{2^3}$. 
 
 A general way of providing examples of $(n,m)$-functions that are $c$-differential bent$_1$, is to take a function $G$ on $\F_{p^n}$ that is perfect $c$-nonlinear (and so, ${_c}D_aF$ is a permutation) for $c$ in a proper subfield $\F_{p^m}$ of $\F_{p^n}$ and apply the relative trace $\Tr_{\F_{p^n}/\F_{p^m}}$ to it, obtaining $F:\F_{p^n}\to \F_{p^m}$ (for some $m$, which is a divisor of $n$) defined by $F(x)=\Tr_{\F_{p^n}/\F_{p^m}}(G(x))$. We now provide the argument. Since $G$ is PcN with respect to $c$, then $G(x+a)-cG(x)$ is a permutation on $\F_{p^n}$, and $\Tr_{\F_{p^n}/\F_{p^m}}\left(G(x+a)-cG(x) \right)$ is therefore balanced on $\F_{p^m}$. Now, using the fact that $c\in\F_{p^m}$, we obtain that
 $\Tr_{\F_{p^n}/\F_{p^m}}\left(G(x+a)-cG(x) \right)={_c}D_aF(x)$, and so, ${_c}D_aF$ is balanced on $\F_{p^m}$.
 We record this as a proposition. 
  \begin{prop}
 Let  $m\,|\,n$, $m<n$, and $p$ prime. If $G$ is PcN on $\F_{p^n}$ with respect to $c\in\F_{p^m}$,  then $F(x)=\Tr_{\F_{p^n}/\F_{p^m}}(G(x))$ is $c$-differential bent$_2$.
 \end{prop}

 We now discuss some of the differential properties of a perfect$_2$ $c$-nonlinear function.
  \begin{thm}
 Let $m,n$ be positive integers and $p$ a prime number, $F\in\cB_{n,p}^m$ and, for all  $0\leq j\leq p^m-1$, we let  $S_{j,c}^u=\{x\in\F_{p^n}\,|\, F(x+u)-cF(x)=j\}$. Then $F$ is a perfect$_2$ $c$-nonlinear function \textup{(}$c$-differential bent$_2$\textup{)} if and only if its output values satisfy $|S_{j+p^{m-1}\ell,c}^u|=|S_{j+p^{m-1}(p-1),c}^u|$, for all $0\leq j\leq p^{m-1}-1$, and $0\leq \ell\leq p-1$.
 \end{thm}
 \begin{proof}
   To show our claim, we order $\F_{p^m}=\{\alpha_0=0,\alpha_1=1,\alpha_2,\ldots,\alpha_{p^m-2}\}$, such that $\sigma(u_j)=j\in\Z_{p^m}$ (the bijective map $\sigma$ was defined in the beginning of this section). With $c\in\F_{p^n}$ constant, for all  $0\leq j\leq p^m-1$, we let  $S_{j,c}^u=\{x\in\F_{p^n}\,|\, F(x+u)-cF(x)=j\}$.   We will use below that the order of the cyclotomic polynomial of index $p^k$ is  $\phi(p^k)=p^{k-1}(p-1)$, for all $k>0$.
  
    If $p=2$ and $u\neq 0$,  since $\zeta_{2^m}^{2^{m-1}+j}=- \zeta_{2^m}^j$, then
 \allowdisplaybreaks
 \begin{align*}
 0={_c}\cC_F(u )
 &=\sum_{x \in \F_{2^n}} \zeta_{2^m}^{F(x+u)  - c  F(x )}=\sum_{j=0}^{2^m-1} |S_{j,c}^u| \zeta_{2^m}^{j}\\
 & =  \sum_{j=0}^{2^{m-1}-1} \left(|S_{j,c}^u|- |S_{j+2^{m-1},c}^u|\right) \zeta_{2^m}^{j},
 \end{align*}
 which will render $|S_{j,c}^u|= |S_{j+2^{m-1},c}^u|$, since $\left\{\zeta_{2^m}^j\,|\,0\leq j\leq 2^{m-1}-1\right\}$ forms a basis for the cyclotomic field $\mathbb{Q}\left(\zeta_{2^m}\right)$.
 
 If $p>2$ and $u\neq 0$, then $\zeta_{p^m}^{\ell p^{m-1}+j}=\zeta_p^\ell \zeta_{p^m}^j$, for $0\leq \ell\leq p-1$, and 
 \allowdisplaybreaks
 \begin{align*}
0={_c}\cC_F(u )
 &=\sum_{x \in \F_{p^n}} \zeta_{p^m}^{F(x+u)  - c  F(x )}=\sum_{j=0}^{p^m-1} |S_{j,c}^u| \zeta_{p^m}^{j}\\
 & =  \sum_{j=0}^{p^{m-1}-1}\left(\sum_{\ell=0}^{p-1}  \zeta_p^\ell |S_{j+p^{m-1}\ell,c}^u| \right) \zeta_{p^m}^{j}.
 \end{align*}
 The  extension $\mathbb{Q}(\zeta_p)\stackrel{p^{m-1}}{\hookrightarrow}\mathbb{Q}(\zeta_{p^m})$ has degree $p^{m-1}$ and the following set $\left\{\zeta_{p^m}^j\,|\,0\leq j\leq p^{m-1}-1\right\}$  forms a basis of $ \mathbb{Q}(\zeta_{p^m})$ over  $\mathbb{Q}(\zeta_p)$, therefore the coefficients in the displayed expression are zero. That is,   for all $0\leq j\leq p^{m-1}-1$,
 \[
 \sum_{\ell=0}^{p-1}  \zeta_p^\ell |S_{j+p^{m-1}\ell,c}^u| =0.
 \]
Again, using that the set $\left\{\zeta_{p}^j\,|\,0\leq j\leq p-2\right\}$ forms a basis for the cyclotomic field $\mathbb{Q}\left(\zeta_{p}\right)$ over $\mathbb{Q}$ and that $\zeta_p^{p-1}=-(1+\zeta_p+\cdots+\zeta_p^{p-2})$, we get
\[
 \sum_{\ell=0}^{p-1}  \zeta_p^\ell \left( |S_{j+p^{m-1}\ell,c}^u|-|S_{j+p^{m-1}(p-1),c}^u|\right) =0,\text{ for all $0\leq j\leq p^{m-1}-1$},
 \]
from which
we infer that $|S_{j+p^{m-1}\ell,c}^u|=|S_{j+p^{m-1}(p-1),c}^u|$, for $0\leq j\leq p^{m-1}-1$, and $0\leq \ell\leq p-1$. 
If the previous condition will hold, by reversing the argument, we find that $F$ is perfect$_2$ $c$-nonlinear.
\end{proof}

\section{Concluding remarks}
\label{sec4}

In this paper we define two different cross/autocorrelations for vectorial $p$-ary $(n,m)$-functions and the corresponding concepts of perfect $c$-nonlinear and $c$-differential bent functions in this  context. We   show that $c$-differential bent functions correspond to perfect $c$-nonlinear functions, thus extending Nyberg's classical result~\cite{Nyb91}. Observe that if $m=1$, the two $c$-differential bent concepts coincide with the classical bent notion~\cite{Ku85}, so the new definitions can be regarded as generalizations in two different directions. We only concentrated here on a few classes of functions (Maiorana-McFarland, Gold, Coulter-Matthews and Dembowski-Ostrom polynomials)  and investigated their $c$-differential bent properties (mostly, for the first bent type). It would be interesting to check other classes of functions for their $c$-differential bent$_1$  or bent$_2$ properties.

 \end{document}